\documentclass[a4paper, aps, pra, reprint, nofootinbib]{revtex4-1}
\usepackage[cm]{fullpage}
\usepackage[T1]{fontenc}
\usepackage{amssymb, amsmath, amsthm}
\makeatletter
\def\amsbb{\use@mathgroup \M@U \symAMSb}
\makeatother
\usepackage[mathscr]{euscript}
\usepackage{mathtools}
\usepackage{verbatim}
\usepackage{bbold}
\usepackage{xcolor}
\definecolor{darkred}{RGB}{200, 0, 0}
\definecolor{darkgreen}{RGB}{0, 100, 0}
\definecolor{darkblue}{RGB}{0, 0, 200}
\usepackage{hyperref}
\hypersetup{colorlinks, breaklinks, linkcolor=darkred, urlcolor=darkblue, citecolor=darkgreen}

\newcommand{\nbox}[2][9]{\hspace{#1pt} \mbox{#2} \hspace{#1pt}}

\newtheorem{prop}{Proposition}[section]

\DeclareMathOperator{\tr}{tr}

\def \diracspacing {0.7pt}

\newcommand{\ket}[1]{| \hspace{\diracspacing} #1 \rangle}

\newcommand{\ketbra}[2]{| \hspace{\diracspacing} #1 \rangle \langle #2 \hspace{\diracspacing} |} 
 
\newcommand{\ketbraq}[1]{\ketbra{#1}{#1}}

\newcommand{\norm}[2][]{#1| \! #1| #2 #1| \! #1|}
\newcommand{\abs}[2][]{#1| #2 #1|}

\newcommand{\cS}{\mathcal{S}}

\newcommand{\comsum}[1]
{
\sum_{x \in \cS_{n}^{#1}} \bigotimes_{j = 1}^{n} \bigg( \frac{ \abs{ [ A_{0}^{j}, A_{1}^{j} ] } }{2} \bigg)^{x_{j}}
}
\newcommand{\biasedCHSH}[2]
{
\begin{equation*}
W_{#1} := #1 ( A_{0} + A_{1} ) \otimes B_{0} + ( A_{0} - A_{1} ) \otimes B_{1}#2
\end{equation*}
}
\newcommand{\biasedsquaredinequality}
{
\begin{equation*}
W_{\alpha}^{2} \leq 2 ( \alpha^{2} + 1 ) \cdot \mathbb{1} \otimes \mathbb{1} + T_{\alpha} \otimes \mathbb{1}
\end{equation*}
for
\begin{equation*}
T_{\alpha} := (\alpha^{2} - 1) \{ A_{0}, A_{1} \} + 2 \alpha \abs{ [ A_{0}, A_{1} ] }.
\end{equation*}
}
\newcommand{\tradeoffchsh}
{
\begin{equation*}
\beta_{\alpha} \leq 2 \sqrt{ \alpha^{2} + t_{\alpha} }.
\end{equation*}
}
\newcommand{\tradeoffmabk}[1]
{
\begin{equation*}
\beta_{n} \leq \sqrt{2^{n - 2}} \cdot \sqrt{ 1 + t_{k} }#1
\end{equation*}
}
\begin{document}
\title{Self-testing of binary observables based on commutation}
\author{J\k{e}drzej Kaniewski}
\affiliation{QMATH, Department of Mathematical Sciences, University of Copenhagen, Universitetsparken 5, 2100 Copenhagen, Denmark}
\date{\today}
\begin{abstract}
We consider the problem of certifying binary observables based on a Bell inequality violation alone, a task known as \emph{self-testing of measurements}. We introduce a family of commutation-based measures, which encode all the distinct arrangements of two projective observables on a qubit. These quantities by construction take into account the usual limitations of self-testing and since they are ``weighted'' by the (reduced) state, they automatically deal with rank-deficient reduced density matrices. We show that these measures can be estimated from the observed Bell violation in several scenarios and the proofs rely only on standard linear algebra. The trade-offs turn out to be tight and, in particular, they give non-trivial statements for arbitrarily small violations. On the other extreme, observing the maximal violation allows us to deduce precisely the form of the observables, which immediately leads to a complete rigidity statement. In particular, we show that for all $n \geq 3$ the $n$-partite Mermin-Ardehali-Belinskii-Klyshko inequality self-tests the $n$-partite Greenberger-Horne-Zeilinger state and maximally incompatible qubit measurements on every party. Our results imply that any pair of projective observables on a qubit can be certified in a truly robust manner. Finally, we show that commutation-based measures give a convenient way of expressing relations among more than two observables.
\end{abstract}
\maketitle
\section{Introduction}The fact that quantum mechanics is incompatible with the concept of local realism~\cite{bell64a} is arguably one of the most surprising features of the quantum world. It should therefore come as no surprise that Bell nonlocality is an attractive field of research for both theoreticians (see~Ref.~\cite{brunner14a} for a review) and experimentalists (see e.g.~the recent loophole-free Bell tests~\cite{hensen15a, giustina15a, shalm15a, hensen16a}). An important practical application of Bell nonlocality is \emph{device-independent quantum cryptography}, whose goal is to prove the security of protocols executed using potentially untrusted devices (see Refs.~\cite{barrett05a, acin06a, colbeck06a, acin07a, colbeck11a} for the early contributions and Ref.~\cite{ekert14a} for a relatively up-to-date review). What makes this task possible is the fact that observing nonlocal correlations allows us to draw conclusions about the inner workings of the untrusted devices. In fact, certain extremal quantum correlations identify exactly the quantum system under consideration (up to well-understood equivalences). For example the only manner to achieve the maximal violation of the Clauser-Horne-Shimony-Holt (CHSH)~\cite{clauser69a} inequality is to perform anticommuting measurements on the maximally entangled state of two qubits~\cite{summers87a, tsirelson93a, popescu92a}. Therefore, not only can we use a Bell test to reject local realism, but if we assume that quantum mechanics provides the correct description of our system, we can essentially identify what it is. Mayers and Yao realised that this allows us to certify quantum devices under minimal assumptions and they also coined the term \emph{self-testing}~\cite{mayers98a, mayers04a}. The general question is simple: ``We have conducted a Bell test and observed certain nonlocal correlations. What can we rigorously deduce about the state shared between the devices and the measurements performed?''.

The first self-testing results only applied in the case of observing the ideal statistics. While interesting from the foundational point of view, this is not sufficient for practical applications, because in a real-world experiment one never observes ideal statistics for at least two reasons: (i) no physical system can be perfectly calibrated and shielded from external noise and (ii) one can only perform a finite number of test runs, so we can only hope to \emph{estimate} the actual probabilities. In order to make statements relevant for experiments we must make them \emph{robust}, i.e.~we have to show that if the observed statistics are close to the ideal ones, then the quantum device should be close (in some well-defined sense) to the perfect realisation. To simplify the problem, instead of looking at the entire probability distribution, we often only look at the violation of some fixed Bell inequality. A Bell inequality is given by a vector of real coefficients $c_{abxy} \in \amsbb{R}$ and the observed Bell value $\beta$ is calculated as
\begin{equation}
\label{eq:bell-inequality}
\beta := \sum_{abxy} c_{abxy} \Pr[a, b | x, y],
\end{equation}
where $\Pr[a, b | x, y]$ is the probability of observing outputs $a, b$ given inputs $x, y$. Let $\beta^{L}$ and $\beta^{Q}$ be the largest values achievable by local-realistic theories and quantum mechanics, respectively, and suppose that $\beta^{L} < \beta^{Q}$. A necessary condition to make a self-testing statement is to observe some violation ($\beta > \beta_{L}$) and a self-testing result is called robust if we can make conclusions even if the violation is not maximal ($\beta < \beta^{Q}$). It is important to distinguish self-testing results which only apply if the violation is close to maximal from those that cover a sizeable portion of the interval $[\beta^{L}, \beta^{Q}]$. The two types can be easily distinguished by writing the observed violation $\beta$ as a convex combination of $\beta^{L}$ and $\beta^{Q}$
\begin{equation*}
\beta = c \beta^{L} + (1 - c) \beta^{Q}
\end{equation*}
and estimating the largest value of $c$ for which we can still make a non-trivial self-testing statement, which we denote $c^{*}$. For all the self-testing results we are aware of we have either $c^{*} \leq 10^{-4}$ or $c^{*} \geq 10^{-1}$. Results of the first kind aim to self-test high-dimensional quantum systems and often have complexity-theoretic implications. Results of the second kind, on the other hand, usually consider simple (low-dimensional) quantum systems and their goal is to derive statements which can be applied to real-world experiments, i.e.~they might actually be useful in designing robust and efficient testing procedures for real devices. Deriving such experimentally-relevant self-testing statements is precisely the focus of this work.

The main challenge in deriving robust self-testing statements lies in finding a natural mathematical formulation of the problem. Since our goal is to make statements even for statistics significantly differing from the ideal setup, we cannot aim for a complete description. We should instead pin down the relevant property and certify precisely that property. This is how our approach differs from the standard formulation, in which one attempts to certify closeness (in trace distance) to the perfect realisation.

The primary goal of this work is to certify two-outcome (binary) projective measurements.\footnote{Since any non-projective measurement can be simulated by a projective measurement on a larger Hilbert space, one can never certify ``non-projectiveness'' of a measurement (unless one imposes an additional dimension bound). This is similar to the well-known fact that mixed states cannot be self-tested.} We propose a novel formulation based on commutation, which recovers several previous results as extreme cases. Commutation-based measures are easily computable, have a simple physical interpretation and demonstrate that all pairs of projective observables on a qubit can be certified in a truly robust fashion. Although, as is common in self-testing and quantum cryptography, we assume that the systems under study are finite-dimensional, as outlined in Appendix~\ref{app:biased-chsh} the results hold essentially unchanged in the infinite-dimensional case.

While self-testing of quantum states has received significant attention in the regime of small~\cite{slofstra11a, mckague12a, yang13a, mckague14a, bamps15a, mckague16a, supic16a, ostrev16a, coladangelo17b, coudron16a, mckague16b, chao16a, natarajan16a, coladangelo17a, chao17a} and experimentally-relevant~\cite{bardyn09a, yang14a, wu14a, pal14a, bancal15a, wu16a, kaniewski16b} robustness, self-testing of measurements is a significantly less studied topic. Although most results in the small robustness regime come as complete rigidity statements (i.e.~they also characterise the optimal measurements), there are only two results exhibiting experimentally-relevant robustness. Bancal et al.~used semidefinite programming to study complementarity of observables in the CHSH case~\cite{bancal15a}, which is a valuable contribution and we see our approach as a refined formulation of the problem which allows for a simple analytic treatment. The second relevant result relates a particular measure of incompatibility between measurements to the observed Bell violation~\cite{cavalcanti16a, chen16a}. This is, however, slightly unsatisfactory, because the quantity considered is a generic measure of incompatibility and cannot be interpreted as a distance to some well-defined ideal arrangement of observables.

\section{Self-testing of observables based on commutation}
The first manner in which we depart from the usual formulation of the self-testing problem, is that we certify observables of one party at a time, i.e.~we have a separate statement for each party, which depends only on the local observables and the reduced state. This is in line with the idea of focusing on a single property (the incompatibility of the observables of Alice), instead of certifying the whole setup (the observables of Alice, the observables of Bob and the shared state) at the same time.

To clarify what kind of statements we can hope for, it is instructive to understand the inherent limitations of self-testing. The two inherent limitations of self-testing (i.e.~properties that cannot be deduced from the outcome statistics) are: the presence of auxiliary degrees of freedom (i.e.~degrees of freedom upon which all measurement operators act trivially) and the application of local unitaries. It is clear that these two equivalences do not affect commutation relations between observables. We might therefore conclude that what we should be certifying is precisely the commutation structure between the observables. This is, however, not quite correct as we can only make statements about the observables \emph{on the support of the (reduced) state}. Therefore, instead of making statements about the observables, we consider scalar quantities of the form $t := \tr (T \rho_{A})$, where $\rho_{A}$ is the reduced state of the subsystem to be measured and $T$ is a Hermitian operator constructed from the observables (whose exact definition depends on the commutation structure we wish to certify).

An appealing feature of these measures is the fact that the maximal value of $t$ is achieved by essentially just one arrangement of observables, i.e.~in the extreme case we can determine the exact form of the observables on the support of the state (up to the aforementioned equivalences). We stress that whenever we make a statement directly about the operators, we implicitly assume that the reduced state is full-rank.

Effective (i.e.~state-dependent) commutation relations are important because they lead to well-understood uncertainty relations~\cite{tomamichel13a, kaniewski14a}. Such \emph{device-independent uncertainty relations}, which certify uncertainty given only limited knowledge about the device, are fundamental as they pin down the exact source of uncertainty in quantum mechanics. Moreover, they have already been used to prove the security of device-independent quantum key distribution~\cite{lim13a} and two-party cryptography~\cite{kaniewski16a, ribeiro16a, ribeiro16b}.

\subsection{Methods}A binary observable is a Hermitian operator $A$ satisfying $- \mathbb{1} \leq A \leq \mathbb{1}$ (or equivalently $A^{2} \leq \mathbb{1}$). An observable is projective if $A^{2} = \mathbb{1}$, but we do not a priori assume this. Since our goal is to identify all the quantum realisations consistent with particular outcome statistics, the projectiveness of the observables and the purity of the state are rigorously \emph{deduced} rather than \emph{assumed}.\footnote{In fact, we can only draw conclusions about the purity of the part of the state upon which the measurements acts non-trivially.}

It is well known that in the case of binary observables commutators and anticommutators of observables appear in the square of the Bell operator. The Bell operator $W$ corresponding to inequality~\eqref{eq:bell-inequality} is defined as
\begin{equation*}
W = \sum_{abxy} c_{abxy} P_{a}^{x} \otimes Q_{b}^{y},
\end{equation*}
where $\{P_{a}^{x}\}_{a}$ is Alice's measurement corresponding to setting $x$ and $\{Q_{b}^{y}\}_{b}$ is Bob's measurement corresponding to setting $y$. It is immediate that the Bell value $\beta$ can be computed as $\beta = \tr (W \rho_{AB})$. Recall that the Cauchy-Schwarz inequality for operators reads $ \abs{ \tr (X^{\dagger} Y) }^{2} \leq \tr (X^{\dagger} X) \cdot \tr (Y^{\dagger} Y)$. Setting $X^{\dagger} = W \sqrt{\rho_{AB}}$ and $Y = \sqrt{\rho_{AB}}$ immediately gives
\begin{equation*}
\beta^{2} = [ \tr (W \rho_{AB}) ]^{2} \leq \tr (W^{2} \rho_{AB}) \cdot \tr \rho_{AB} = \tr (W^{2} \rho_{AB}).
\end{equation*}
Therefore, proving an operator inequality
\begin{equation}
\label{eq:main-inequality}
W^{2} \leq g( A_{0}, A_{1} ) \otimes \mathbb{1},
\end{equation}
where $A_{0}$ and $A_{1}$ are the observables of Alice and $g$ is a function which outputs a Hermitian operator, immediately implies
\begin{equation*}
\beta \leq \sqrt{ \tr \big( g( A_{0}, A_{1} ) \rho_{A} \big) }.
\end{equation*}
If the right-hand side provides a useful characterisation of the observables of Alice, this constitutes a self-testing statement. Let us stress that for projective observables $W^{2}$ can often be written explicitly as a function of their commutators and anticommutators, which provides helpful intuition on the possible form of the function $g(A_{0}, A_{1})$. In the remainder of this paper we apply this idea to specific Bell scenarios.

\subsection{Certifying anticommuting observables}In the CHSH scenario Alice and Bob measure one of two binary observables denoted $A_{j}$ and $B_{k}$ for $j, k \in \{0, 1\}$. The CHSH operator is defined as
\biasedCHSH{}{}
for which $\beta^{L} = 2$ and $\beta^{Q} = 2 \sqrt{2}$. In Appendix~\ref{app:biased-chsh} we prove that
\begin{equation*}
W^{2} \leq 4 \cdot \mathbb{1} \otimes \mathbb{1} - [A_{0}, A_{1}] \otimes [B_{0}, B_{1}]
\end{equation*}
and by noting that $\abs{ [B_{0}, B_{1}] } \leq 2 \cdot \mathbb{1}$ we obtain
\begin{equation*}
W^{2} \leq 4 \cdot \mathbb{1} \otimes \mathbb{1} + 2 \, \abs{ [A_{0}, A_{1}] } \otimes \mathbb{1}.
\end{equation*}
From the argument outlined above we deduce that
\begin{equation}
\label{eq:chsh-commutator}
\beta \leq 2 \sqrt{ 1 + t },
\end{equation}
where $t := \frac{1}{2} \tr \big( \abs{ [A_{0}, A_{1}] } \rho_{A} \big) \in [0, 1]$ is the \emph{effective commutator}. This scalar quantity is invariant under local unitaries and adding extra degrees of freedom, it avoids making any statement about the observables outside the support of $\rho_{A}$ and is easily computable. The physical interpretation is clear: $t$ measures the incompatibility of Alice's observables ``weighted'' by the reduced state $\rho_{A}$. The matrix modulus, which arises in the derivation, avoids cancellations, e.g.~$t = 0$ implies that the observables commute on the support of $\rho_{A}$, which prevents us from observing any violation. At the other extreme, the maximal value $t = 1$ implies the existence of a unitary $U_{A}$ (see Proposition~\ref{prop:biased-chsh-optimal-observables} in Appendix~\ref{app:biased-chsh}) such that
\begin{align}
\label{eq:anticommuting-observables}
A_{0} = U_{A} ( \sigma_{x} \otimes \mathbb{1} ) U_{A}^{\dagger} \nbox{and} A_{1} = U_{A} ( \sigma_{y} \otimes \mathbb{1} ) U_{A}^{\dagger}
\end{align}
(recall the assumption that $\rho_{A}$ is full-rank). This shows that $t$ is a useful measure of how close Alice's observables are to a pair of anticommuting observables on a qubit. Inequality~\eqref{eq:chsh-commutator} is interesting for several reasons: it gives a non-trivial statement as soon as $\beta > 2$, it is tight (see Proposition~\ref{prop:biased-chsh-tight-trade-off} in Appendix~\ref{app:biased-chsh}) and observing the maximal violation $\beta = 2 \sqrt{2}$ implies $t = 1$, which allows us to deduce the exact form of the observables. Note that a similar approach has been previously used to analyse a semi-device-independent prepare-and-measure scenario where the transmitted system is a qubit~\cite{woodhead13a, woodhead14a}.

Although our primary goal is to certify observables, in the case of perfect statistics this argument leads to a complete rigidity statement. If $\beta = 2 \sqrt{2}$, then by symmetry the observables of Bob satisfy a relation analogous to eq.~\eqref{eq:anticommuting-observables} for some other unitary $U_{B}$ and, therefore, $W$ up to the unitary $U := U_{A} \otimes U_{B}$ is simply a two-qubit operator tensored with identities. The structure of this two-qubit operator is well known and, in particular, the largest eigenvalue is non-degenerate. The corresponding eigenstate is maximally entangled and we denote its density matrix $\Phi$. It is now clear that any state satisfying $\tr (W \rho_{AB}) = 2 \sqrt{2}$ must be of the form
\begin{equation*}
\rho_{AB} = U ( \Phi_{AB} \otimes \sigma_{A'B'} ) U^{\dagger},
\end{equation*}
where $\sigma_{A'B'}$ is an arbitrary state on the extra degrees of freedom. It is not hard to see that this rigidity argument is in spirit quite close to the original approach of Popescu and Rohrlich~\cite{popescu92a}.

\subsection{Certifying anticommuting observables in the multipartite case}This approach can be generalised to the multipartite setting. We consider an $n$-partite ($n \geq 2$) scenario in which every party has two binary observables and we denote the observables of the $k$-th party $A^{k}_{0}$ and $A^{k}_{1}$. The Bell operator $W_{n}$ for the Mermin-Ardehali-Belinskii-Klyshko (MABK) \cite{mermin90a, ardehali92a, belinskii93a} family is defined recursively:
\begin{equation}
\label{eq:mabk-definition}
W_{n} := \frac{1}{2} W_{n - 1} \otimes \big( A_{0}^{n} + A_{1}^{n} \big) + \frac{1}{2} W_{n - 1}' \otimes \big( A_{0}^{n} - A_{1}^{n} \big),
\end{equation}
where $W_{1} = A_{0}^{1}$ and the primed operators are obtained by swapping $A_{0}^{j} \leftrightarrow A_{1}^{j}$ for all $j$~\cite{belinskii93a}. It is easy to check that up to normalisation $W_{2}$ and $W_{3}$ correspond to the CHSH and Mermin operators, respectively. In Appendix~\ref{app:mabk} we show that
\begin{equation}
\label{eq:mabk-upper-bound}
W_{n}^{2} \leq \comsum{e},
\end{equation}
where $\cS_{n}^{e}$ is the set of $n$-bit strings of even parity, $x_{j}$ is the $j$-th bit of $x$ and we use the convention $X^{1} := X$, $X^{0} := \mathbb{1}$. This operator inequality captures the relation between the Bell violation and the local incompatibility of observables, e.g.~forcing all the observables to commute ($[A_{0}^{j}, A_{1}^{j}] = 0$ for all $j$) recovers the local-realistic bound $\beta_{n}^{L} = 1$ ($W_{n}^{2} \leq \mathbb{1}$), while the trivial upper bound $\abs{ [ A_{0}^{j}, A_{1}^{j} ] } \leq 2 \cdot \mathbb{1}$ correctly identifies the maximal quantum violation $\beta_{n}^{Q} = \sqrt{ 2^{n - 1} }$ ($W_{n}^{2} \leq 2^{n - 1} \cdot \mathbb{1}$). To certify the observables of the $k$-th party, we place the trivial upper bound on all the other commutators to obtain
\begin{equation*}
W_{n}^{2} \leq 2^{n - 2} \; \mathbb{1} \otimes \Big( \mathbb{1} + \frac{1}{2} \abs{ [ A_{0}^{k}, A_{1}^{k} } \Big).
\end{equation*}
This implies that
\tradeoffmabk{,}
where $t_{k}$ is the effective commutator of the $k$-th party. Just like in the CHSH case this trade-off is tight (see Proposition~\ref{prop:mabk-tight-trade-off} in Appendix~\ref{app:mabk}) and, in particular, the upper bound corresponding to $t_{k} = 0$ correctly identifies the largest violation achievable if one of the parties acts classically. In order to achieve the optimal violation $\beta_{n} = \beta_{n}^{Q}$, we require $t_{j} = 1$ for all $j$, which implies that the observables of each party take the form given by eq.~\eqref{eq:anticommuting-observables}. Through an argument analogous to the CHSH case we identify the subspace corresponding to the maximal violation and conclude that the state must be of the form
\begin{equation*}
\rho_{QR} = U_{QR} ( \Gamma_{Q_{1} \cdots Q_{n}} \otimes \sigma_{R_{1} \cdots R_{n}} ) U_{QR}^{\dagger},
\end{equation*}
where $Q_{k} R_{k}$ denotes the Hilbert space of the $k$-th party, $U_{QR} = \bigotimes_{j = 1}^{n} U_{Q_{j} R_{j}}^{j}$ is a product unitary, $\Gamma$ is a $n$-partite Greenberger-Horne-Zeilinger (GHZ)~\cite{greenberger89a} state and $\sigma_{R_{1} \cdots R_{n}}$ is an arbitrary state on the additional degrees of freedom. To the best of our knowledge this complete rigidity statement for the entire family of MABK inequalities constitutes a new result.\footnote{The case of $n = 3$ was fully analysed by Colbeck~\cite{colbeck06a} and later Miller and Shi~\cite{miller13a}, while robust state certification for $n = 3$ and $n = 4$ was shown numerically by P{\'a}l et al.~\cite{pal14a}.}

\subsection{Certifying non-maximally incompatible observables}In the previous cases the optimal observables on every party correspond to anticommuting observables on a qubit. Here, we show that an arbitrary pair of qubit observables, not necessarily maximally incompatible, is exactly characterised through their commutation relation. For $\alpha \geq 1$ we consider the generalisation of the CHSH inequality:
\biasedCHSH{\alpha}{}
introduced by Lawson, Linden and Popescu~\cite{lawson10a} and later used to investigate the relation between nonlocality and randomness~\cite{acin12a}. The local-realistic and quantum bounds for this inequality equal $\beta_{\alpha}^{L} = 2 \alpha$ and $\beta_{\alpha}^{Q} = 2 \sqrt{ \alpha^{2} + 1 }$ (hence $\beta_{\alpha}^{L} < \beta_{\alpha}^{Q}$ for all $\alpha \geq 1$) and the maximal violation is achieved by measuring a maximally entangled two-qubit state, but the optimal observables of Alice are no longer maximally incompatible. In Appendix~\ref{app:biased-chsh} we show that
\biasedsquaredinequality
Defining $t_{\alpha} := \frac{1}{4} \tr ( T_{\alpha} \rho_{A} ) - \frac{1}{2} (\alpha^{2} - 1)$, which recovers the effective commutator for $\alpha = 1$, allows us to write
\tradeoffchsh
If the observables of Alice commute, we have $t_{\alpha} \leq 0$, which immediately recovers the classical bound. On the other hand, observing the maximal violation $\beta_{\alpha} = \beta_{\alpha}^{Q}$ implies that $t_{\alpha} = 1$ and that there exists a unitary $U_{A}$ (see Proposition~\ref{prop:biased-chsh-optimal-observables} in Appendix~\ref{app:biased-chsh}) such that
\begin{gather*}
A_{0} = U_{A} ( \sigma_{x} \otimes \mathbb{1} ) U_{A}^{\dagger},\\
A_{1} = U_{A} \big( [ \cos \theta_{\alpha} \, \sigma_{x} + \sin \theta_{\alpha} \, \sigma_{y} ] \otimes \mathbb{1} \big) U_{A}^{\dagger}
\end{gather*}
for $\theta_{\alpha} := \arccos \big( \frac{ \alpha^{2} - 1 }{ \alpha^{2} + 1 } \big) \in (0, \pi/2]$. This characterises the exact commutation structure between the observables and by considering $\alpha \in [1, \infty)$ we can certify any angle between two projective observables on a qubit (angles larger than $\pi/2$ correspond to simply relabelling the outcomes of one of the observables). The maximal violation is only possible if the observables of Bob anticommute, which leads directly to a rigidity statement for the generalised CHSH inequality. It is worth pointing out that the statistics that maximally violate this inequality belong to the family of self-tests recently derived by Wang, Wu and Scarani~\cite{wang16a}.

\subsection{Certifying multiple anticommuting observables}
We conclude by showing that commutation-based measures are useful also in multiobservable scenarios. The simplest arrangement of three binary observables is arguably given by the three Pauli matrices:~$( \sigma_{x}, \sigma_{y}, \sigma_{z} )$. However, as observed by McKague and Mosca such an arrangement cannot be distinguished from $( \sigma_{x}, - \sigma_{y}, \sigma_{z} )$ in a device-independent fashion~\cite{mckague11a}. Up to this equivalence they showed that the relation among the three observables can be deduced from the pairwise relations between $(\sigma_{x}, \sigma_{y})$, $(\sigma_{x}, \sigma_{z})$ and $(\sigma_{y}, \sigma_{z})$. More specifically, if every pair satisfies the Mayers-Yao self-testing criterion, then there exists a unitary $U_{A}$ and a projective observable $\Upsilon$ (i.e.~$\Upsilon^{2} = \mathbb{1}$) such that
\begin{align}
\label{eq:three-observables}
A_{0} = U_{A} ( &\sigma_{x} \otimes \mathbb{1} ) U_{A}^{\dagger},\nonumber\\
A_{1} = U_{A} ( &\sigma_{y} \otimes \mathbb{1} ) U_{A}^{\dagger},\\
A_{2} = U_{A} ( &\sigma_{z} \otimes \Upsilon ) U_{A}^{\dagger}\nonumber.
\end{align}
What is slightly unsatisfactory about this statement, is the inequivalent treatment of the three observables: the additional observable $\Upsilon$ appears only in the equation for $A_{2}$.\footnote{It is clear that $\Upsilon$ can be placed in any of the three observables, but this will not lead to a symmetric statement.} On the other hand, we set out to certify three equivalent observables and the testing procedure treated all of them on equal footing. This aesthetic problem can be avoided if one uses a commutation-based approach. To self-test $n$ anticommuting observables we use an extension of the CHSH game introduced by Slofstra~\cite{slofstra11a} and observing the maximal violation implies that
\begin{equation}
\tr \big( \abs{ [A_{j}, A_{k}] } \rho_{A} \big) = 2
\end{equation}
for all $j \neq k$. In Appendix~\ref{app:three-observables} we show that for three observables, these pairwise conditions recover precisely the statement of McKague and Mosca, i.e.~they constitute a compact and symmetric formulation of relation~\eqref{eq:three-observables}. The argument can be easily extended to an arbitrary number of observables: for even $n$ we derive the exact form of the observables (which coincides with the standard construction for multiple anticommuting observables), while for odd $n$ we specify the form up to a single projective observable (which plays the same role as $\Upsilon$ for $n = 3$).

\section{Conclusions}In this paper we have proposed a novel commutation-based formulation to quantify relations between binary observables and derived nonlocality-incompatibility trade-offs by simple algebraic manipulations of the Bell operator. Our measures are mathematically convenient, easily computable and come with a natural measure of closeness and automatically deal with the case of rank-deficient reduced states. All the trade-offs derived in this paper are tight and the incompatibility measures are strict enough to recover the exact form of the optimal observables, but at the same time they yield non-trivial statements for arbitrary small violations. Since observing the maximal violation allows us to deduce precisely the form of the observables (on the support of the state), we immediately obtain complete rigidity statements. We have shown that any arrangement of two projective observables on a qubit is uniquely characterised by an easily computable commutation-based quantity and that these quantities are also useful in expressing relations among more than two observables.

In this work we have considered Bell inequalities which are maximally violated by projective measurements on qubits. A natural extension of this work would be to consider inequalities with binary outcomes which cannot be maximally violated by qubits, e.g.~the famous $I_{3322}$ inequality. Since this inequality is believed to be maximally violated only by infinite-dimensional systems~\cite{pal10a}, we conjecture that the maximal violation requires a commutation structure unachievable in any finite dimension. Clearly, this approach can be reversed and by finding commutation structures which cannot be realised in finite dimensions, one can hope to construct Bell inequalities whose maximal violation requires infinite dimensions. Another direction is to generalise this approach to measurements with more outcomes. For instance the Heisenberg-Weyl observables in dimension $d$ satisfy the ``twisted'' commutation relation $Z_{d} X_{d} = \omega X_{d} Z_{d}$, where $\omega = \exp ( 2 \pi i/d )$ and it is known that for $d = 3$ the corresponding measurements (which are easily seen to be mutually unbiased) achieve the maximal violation of a certain Bell inequality~\cite{buhrman05a, ji08a, liang09a}. Does the maximal violation certify precisely this commutation structure? What about similar inequalities tailored for higher-dimensional systems? Another choice of incompatible qudit measurements are the Collins-Gisin-Linden-Massar-Popescu (CGLMP) measurements and these violate maximally several inequalities~\cite{collins02a, devicente15a, salavrakos17a}. What is the commutation structure realised by the CGLMP measurements? Is it certified by the maximal violation of these inequalities?

\section{Acknowledgements}We would like to thank Yeong-Cherng Liang, Chris Perry and P{\'e}ter Vrana for constructive comments on an early version of this manuscript. We acknowledge fruitful discussions with Antonio Ac{\'i}n, Remik Augusiak, Nicolas Brunner, Matthias Christandl, Laura Man\v{c}inska, Alex M{\"u}ller-Hermes, Paul Skrzypczyk and Erik Woodhead. We are particularly grateful to Giacomo De Palma for help regarding the infinite-dimensional case. We acknowledge funding from the European Union's Horizon 2020 research and innovation programme under the Marie Sk{\l}odowska-Curie Action ROSETTA (grant no.~749316), the European Research Council (grant no.~337603), the Danish Council for Independent Research (Sapere Aude) and VILLUM FONDEN via the QMATH Centre of Excellence (grant no.~10059).
\appendix
\section{The biased CHSH inequality}
\label{app:biased-chsh}
In this Appendix we analyse the family of biased CHSH inequalities parametrised by $\alpha \in [1, \infty)$. The actual CHSH inequality corresponds to setting $\alpha = 1$.

Let $A_{j}$ and $B_{k}$ be binary observables, i.e.~Hermitian operators satisfying $- \mathbb{1} \leq A_{j}, B_{k} \leq \mathbb{1}$, and let us for now assume that they act on a finite dimensional Hilbert space. We define the Bell operator $W_{\alpha}$ as
\biasedCHSH{\alpha}{.}
\begin{prop}
\label{prop:biased-chsh-upper-bound}
The square of the Bell operator satisfies
\biasedsquaredinequality
\end{prop}
\begin{proof}
Multiplying the terms out gives
\begin{align*}
W_{\alpha}^{2} &= \alpha^{2} ( A_{0}^{2} + A_{1}^{2} + \{ A_{0}, A_{1} \} ) \otimes  B_{0}^{2}\\
&+ ( A_{0}^{2} + A_{1}^{2} - \{ A_{0}, A_{1} \} ) \otimes B_{1}^{2}\\
&+ \alpha ( A_{0}^{2} - A_{1}^{2} ) \otimes \{ B_{0}, B_{1} \} - \alpha [ A_{0}, A_{1} ] \otimes [ B_{0}, B_{1} ].
\end{align*}
Using the fact that $B_{k}^{2} \leq \mathbb{1}$ and regrouping terms we obtain
\begin{align*}
W_{\alpha}^{2} &\leq A_{0}^{2} \otimes \big[ ( \alpha^{2} + 1 ) \cdot \mathbb{1} + \alpha \{ B_{0}, B_{1} \} \big]\\
&+ A_{1}^{2} \otimes \big[ ( \alpha^{2} + 1 ) \cdot \mathbb{1} - \alpha \{ B_{0}, B_{1} \} \big] \\
&+ ( \alpha^{2} - 1 ) \{ A_{0}, A_{1} \} \otimes \mathbb{1} - \alpha [ A_{0}, A_{1} ] \otimes [ B_{0}, B_{1} ].
\end{align*}
Since $- 2 \cdot \mathbb{1} \leq \{ B_{0}, B_{1} \} \leq 2 \cdot \mathbb{1}$, we have
\begin{equation*}
( \alpha^{2} + 1 ) \cdot \mathbb{1} \pm \alpha \{ B_{0}, B_{1} \} \geq 0,
\end{equation*}
so we can use $A_{j}^{2} \leq \mathbb{1}$ to obtain
\begin{align*}
W_{\alpha}^{2} \leq 2 ( \alpha^{2} + 1 ) \cdot \mathbb{1} \otimes \mathbb{1} &+ (\alpha^{2} - 1) \cdot \{ A_{0}, A_{1} \} \otimes \mathbb{1}\\
&- \alpha [ A_{0}, A_{1} ] \otimes [ B_{0}, B_{1} ].
\end{align*}
Upperbounding the commutators by their matrix moduli leads to
\begin{align*}
W_{\alpha}^{2} \leq 2 ( \alpha^{2} + 1 ) \cdot \mathbb{1} \otimes \mathbb{1} &+ (\alpha^{2} - 1) \cdot \{ A_{0}, A_{1} \} \otimes \mathbb{1}\\
&+ \alpha \abs{ [ A_{0}, A_{1} ] } \otimes \abs{ [ B_{0}, B_{1} ] }.
\end{align*}
Using $\abs{ [B_{0}, B_{1} ] } \leq 2 \cdot \mathbb{1}$ completes the proof.
\end{proof}
In the next proposition we use the fact that for binary observables $A_{0}, A_{1}$ we have
\begin{equation}
\label{eq:com-anticom}
\abs{ \{A_{0}, A_{1}\} }^{2} + \abs{ [A_{0}, A_{1}] }^{2} \leq 4 \cdot \mathbb{1},
\end{equation}
where $\abs{X} := \sqrt{ X^{\dagger} X }$. The inequality holds as an equality iff both observables are projective (see Section 2.2 of Ref.~\cite{kaniewski16a} for an elementary proof).
\begin{prop}
\label{prop:Talpha-bounds}
The operator $T_{\alpha}$ satisfies
\begin{equation*}
T_{\alpha} \leq 2 ( \alpha^{2} + 1 ) \cdot \mathbb{1}.
\end{equation*}
\end{prop}
\begin{proof}
We first upperbound the anticommutator by its matrix modulus
\begin{equation*}
T_{\alpha} \leq (\alpha^{2} - 1) \abs{ \{ A_{0}, A_{1} \} } + 2 \alpha \abs{ [ A_{0}, A_{1} ] }.
\end{equation*}
Rearranging inequality~\eqref{eq:com-anticom} gives
\begin{equation*}
\abs{ \{A_{0}, A_{1}\} }^{2} \leq 4 \cdot \mathbb{1} - \abs{ [A_{0}, A_{1}] }^{2}.
\end{equation*}
Since the function $f(t) = \sqrt{t}$ is operator monotone, we have
\begin{equation*}
\abs{ \{A_{0}, A_{1}\} } \leq \sqrt{ 4 \cdot \mathbb{1} - \abs{ [A_{0}, A_{1}] }^{2} },
\end{equation*}
which immediately implies that
\begin{equation*}
T_{\alpha} \leq (\alpha^{2} - 1) \sqrt{ 4 \cdot \mathbb{1} - \abs{ [A_{0}, A_{1}] }^{2} } + 2 \alpha \abs{ [ A_{0}, A_{1} ] }.
\end{equation*}
It is clear that the spectrum of the right-hand side is the image of the spectrum of $\abs{[A_{0}, A_{1}]}$ under the function
\begin{equation*}
g(\lambda) := (\alpha^{2} - 1) \sqrt{ 4 - \lambda^{2} } + 2 \alpha \lambda.
\end{equation*}
It is easy to check that the maximum value of $g$ in the interval $[0, 2]$ equals $2 (\alpha^{2} + 1)$ and occurs only for $\lambda = 4 \alpha/(\alpha^{2} + 1)$.
\end{proof}
Our goal now is to show that if we observe the maximal violation, then we can exactly characterise the observables. One of the technical tools used in the proofs is the H{\"o}lder inequality, which states that for arbitrary operators $X, Y$ we have
\begin{equation*}
\abs{ \tr ( X Y ) } \leq \norm{X}_{\infty} \cdot \norm{Y}_{1},
\end{equation*}
where $\norm{\cdot}_{p}$ denotes the Schatten $p$-norm. In our case the operators $X$ and $Y$ are Hermitian, which allows us to drop the absolute value on the left-hand side and, moreover, $Y$ is positive semidefinite, which allows us to replace the $1$-norm with the trace, i.e.~we obtain
\begin{equation*}
\tr ( X Y ) \leq \norm{X}_{\infty} \cdot \tr Y.
\end{equation*}
It is easy to see that if the inequality holds as an equality and $Y$ is full-rank, then
\begin{equation*}
X = \norm{X}_{\infty} \cdot \mathbb{1}.
\end{equation*}
\begin{prop}
\label{prop:biased-chsh-optimal-observables}
Suppose that $\tr( T_{\alpha} \rho_{A} ) = 2 ( \alpha^{2} + 1 )$ and $\rho_{A}$ is full-rank. Then, there exists a unitary $U_{A}$ such that
\begin{align*}
A_{0} &= U_{A} ( \sigma_{x} \otimes \mathbb{1} ) U_{A}^{\dagger},\\
A_{1} &= U_{A} ( \cos \theta_{\alpha} \, \sigma_{x} + \sin \theta_{\alpha} \, \sigma_{y} \otimes \mathbb{1} ) U_{A}^{\dagger}
\end{align*}
for $\theta_{\alpha} := \arccos \big( \frac{ \alpha^{2} - 1 }{ \alpha^{2} + 1 } \big)$.
\end{prop}
\begin{proof}
Since the H{\"o}lder inequality is tight and the reduced state is full-rank we immediately see that
\begin{equation*}
T_{\alpha} = 2 ( \alpha^{2} + 1 ) \cdot \mathbb{1}.
\end{equation*}
This implies that all the steps used to prove the upper bound in Proposition~\ref{prop:Talpha-bounds} are tight. In particular, the observables must be projective $A_{j}^{2} = \mathbb{1}$ and the anticommutator is a positive semidefinite operator $\{ A_{0}, A_{1} \} = \abs{ \{ A_{0}, A_{1} \} }$. Since
\begin{equation*}
\abs{ [A_{0}, A_{1} ] } = \frac{ 4 \alpha }{ \alpha^{2} + 1 } \cdot \mathbb{1}
\end{equation*}
and $[A_{0}, A_{1}]$ is an anti-Hermitian operator, the spectral decomposition reads
\begin{equation*}
[A_{0}, A_{1}] = \frac{4 \alpha  \cdot i}{\alpha^{2} + 1} ( P_{+} - P_{-} )
\end{equation*}
for some orthogonal projectors $P_{+}$ and $P_{-}$ such that $P_{+} + P_{-} = \mathbb{1}$. Since
\begin{equation*}
A_{0} [A_{0}, A_{1}] A_{0} = - [A_{0}, A_{1}],
\end{equation*}
we conclude that $A_{0} P_{\pm} A_{0} = P_{\mp}$. If $\{ \ket{e_{j}^{0}} \}_{j}$ is an orthonormal basis for the support of $P_{+}$, then $\{ \ket{e_{j}^{1}} \}_{j}$ for $\ket{e_{j}^{1}} := A_{0} \ket{e_{j}^{0}}$ is an orthonormal basis for the support of $P_{-}$. Define the unitary $U_{0}$ as
\begin{equation*}
U_{0} \ket{e_{j}^{b}} = \ket{b} \ket{j}
\end{equation*}
for $b \in \{0, 1\}$. It is easy to verify that 
\begin{equation*}
U_{0} [A_{0}, A_{1} ] U_{0}^{\dagger} = \frac{ 4 \alpha \cdot i }{\alpha^{2} + 1} \, \sigma_{z} \otimes \mathbb{1}.
\end{equation*}
Since $\{ \mathbb{1}, \sigma_{x}, \sigma_{y}, \sigma_{z} \}$ constitute an operator basis for linear operators acting on $\amsbb{C}^{2}$, we can without loss of generality write
\begin{equation}
\label{eq:k0-kx-ky-kz}
U_{0} A_{0} U_{0}^{\dagger} = \mathbb{1} \otimes K_{0} + \sigma_{x} \otimes K_{x} + \sigma_{y} \otimes K_{y} + \sigma_{z} \otimes K_{z}
\end{equation}
for some Hermitian operators $K_{0}, K_{x}, K_{y}, K_{z}$. It is easy to check that for projective observables
\begin{equation*}
\{ A_{0}, [ A_{0}, A_{1} ] \} = 0,
\end{equation*}
which implies that $K_{0} = K_{z} = 0$. Requiring that $A_{0}^{2} = \mathbb{1}$ implies that
\begin{equation}
\label{eq:kx-ky-conditions}
K_{x}^{2} + K_{y}^{2} = \mathbb{1} \nbox{and} [K_{x}, K_{y}] = 0.
\end{equation}
Therefore, $K_{x}$ and $K_{y}$ can be simultaneously diagonalised and the eigenvalues can be written in terms of trigonometric functions, i.e.
\begin{equation}
\label{eq:kx-ky}
K_{x} = \sum_{j} \cos \gamma_{j} \, \ketbraq{j} \nbox{and} K_{y} = \sum_{j} \sin \gamma_{j} \, \ketbraq{j}
\end{equation}
for some angles $\gamma_{j}$ and some orthonormal basis $\{ \ket{j} \}$. This allows us to write
\begin{align*}
U_{0} A_{0} U_{0}^{\dagger} &= \sigma_{x} \otimes K_{x} + \sigma_{y} \otimes K_{y}\\
&= \sum_{j} ( \cos \gamma_{j} \, \sigma_{x} + \sin \gamma_{j} \, \sigma_{y} ) \otimes \ketbraq{j}.
\end{align*}
Now we align all the qubit observables using the following controlled unitary
\begin{equation}
\label{eq:u1}
U_{1} = \sum_{j} \exp \Big( \frac{i \gamma_{j}}{2} \cdot \sigma_{z} \Big) \otimes \ketbraq{j}.
\end{equation}
It is easy to check that
\begin{align}
\label{eq:u1-u0-a0}
U_{1} U_{0} A_{0} U_{0}^{\dagger} U_{1}^{\dagger} &= \sigma_{x} \otimes \mathbb{1},\\
U_{1} U_{0} [A_{0}, A_{1} ] U_{0}^{\dagger} U_{1}^{\dagger} &= \frac{ 4 \alpha \cdot i }{\alpha^{2} + 1} \sigma_{z} \otimes \mathbb{1}\nonumber.
\end{align}
An analogous reasoning applied to $A_{1}$ leads to the conclusion that we can without loss of generality write
\begin{equation*}
U_{1} U_{0} A_{1} U_{0}^{\dagger} U_{1}^{\dagger} = \sigma_{x} \otimes K_{x}' + \sigma_{y} \otimes K_{y}'.
\end{equation*}
Since the anticommutator is positive semidefinite and the observables are projective, we know that
\begin{align*}
\{ A_{0}, A_{1} \} = \abs{ \{ A_{0}, A_{1} \} } &= \sqrt{ 4 \cdot \mathbb{1} - \abs{ [A_{0}, A_{1}] }^{2} }\\
&= \frac{ 2 (\alpha^{2} - 1) }{\alpha^{2} + 1} \cdot \mathbb{1}.
\end{align*}
Define $\theta_{\alpha} := \arccos \big( \frac{ \alpha^{2} - 1 }{ \alpha^{2} + 1 } \big) \in [0, \pi/2]$. Imposing consistency on the anticommutator gives
\begin{equation*}
K_{x}' = \cos \theta_{\alpha} \cdot \mathbb{1}.
\end{equation*}
On the other hand, requiring consistency on the commutator implies that
\begin{equation*}
K_{y}' = \sin \theta_{\alpha} \cdot \mathbb{1}.
\end{equation*}
Clearly, setting $U_{A} = U_{0}^{\dagger} U_{1}^{\dagger}$ concludes the proof.
\end{proof}
It is worth pointing out that while in the proof we have (purposefully) avoided using Jordan's lemma, one can clearly see its emergence in eq.~\eqref{eq:kx-ky}. Since $K_{x}$ and $K_{y}$ commute, we can find a common eigenbasis and these eigenstates can be seen as labels for the distinct two-dimensional subspaces.

Let us also mention that all these results easily generalise to the case of infinite-dimensional systems. Suppose that the subsystems of Alice and Bob are described by separable Hilbert spaces. We first observe that the Hermitian operators $K_{0}, K_{x}, K_{y}, K_{z}$ appearing in eq.~\eqref{eq:k0-kx-ky-kz} must be bounded, because $A_{0}$ is bounded. The existence of a basis in which $K_{x}$ and $K_{y}$ are simultaneously diagonal used in eq.~\eqref{eq:kx-ky} is replaced by the existence of a joint spectral resolution. More specifically, according to Theorem 1 in Section 6.5 of Ref.~\cite{birman87a} the commutation condition $[K_{x}, K_{y}] = 0$ implies the existence of a spectral measure $E$ on $\amsbb{R}^{2}$ such that
\begin{equation*}
K_{x} = \int_{\amsbb{R}^{2}} u \; d E(u, v) \nbox{and} K_{y} = \int_{\amsbb{R}^{2}} v \; d E(u, v),
\end{equation*}
which immediately leads to
\begin{align*}
U_{0} A_{0} U_{0}^{\dagger} &= \sigma_{x} \otimes K_{x} + \sigma_{y} \otimes K_{y}\\
&= \int_{\amsbb{R}^{2}} ( u\, \sigma_{x} + v \, \sigma_{y} ) \otimes d E(u, v).
\end{align*}
The condition $K_{x}^{2} + K_{y}^{2} = \mathbb{1}$ implies that the measure $E$ is supported only on the circle $u^{2} + v^{2} = 1$, so we can find a measurable function $\gamma : \amsbb{R}^{2} \to [0, 2 \pi]$ such that for every point $(u, v)$ in the support of $E$ we have $u = \cos \gamma(u,v)$ and $v = \sin \gamma(u,v)$. The unitary $U_{1}$ defined in eq.~\eqref{eq:u1} is replaced with
\begin{equation}
U_{1} = \int_{\amsbb{R}^{2}} \exp \bigg( \frac{i \gamma(u, v)}{2} \cdot \sigma_{z} \bigg) \otimes d E(u, v),
\end{equation}
which immediately recovers eq.~\eqref{eq:u1-u0-a0}:
\begin{equation*}
U_{1} U_{0} A_{0} U_{0}^{\dagger} U_{1}^{\dagger} = \sigma_{x} \otimes \int_{\amsbb{R}^{2}} d E(u, v) = \sigma_{x} \otimes \mathbb{1}.
\end{equation*}
The rest of the proof remains unchanged.
\begin{prop}
\label{prop:biased-chsh-tight-trade-off}
The trade-off between the Bell violation $\beta_{\alpha}$ and the measure of incompatibility on Alice's side $t_{\alpha}$
\tradeoffchsh
is tight for $t_{\alpha} \in [0, 1]$.
\end{prop}
\begin{proof}
The proof crucially relies on the fact that many of the operator inequalities used in proving Proposition~\ref{prop:biased-chsh-upper-bound} turn out to be tight for projective observables. Suppose that Alice and Bob share a two-qubit state and perform the measurements
\begin{gather*}
A_{0} = \sigma_{x}, \quad A_{1} = \cos \gamma \, \sigma_{x} + \sin \gamma \, \sigma_{y},\\
B_{0} = \sigma_{x}, \quad B_{1} = \sigma_{y}
\end{gather*}
for $\gamma \in [0, \pi/2]$. It is easy to check that
\begin{equation*}
W_{\alpha}^{2} = 2 \big[ ( \alpha^{2} + 1 ) + (\alpha^{2} - 1) \cos \gamma \big] \cdot \mathbb{1} \otimes \mathbb{1} + 4 \alpha \sin \gamma \, \sigma_{z} \otimes \sigma_{z}.
\end{equation*}
It is easy to see that the largest eigenvalue of $W_{\alpha}^{2}$ equals
\begin{equation*}
\lambda_{\gamma} := 2 \big[ ( \alpha^{2} + 1 ) + (\alpha^{2} - 1) \cos \gamma \big] + 4 \alpha \sin \gamma.
\end{equation*}
Since
\begin{equation*}
( \sigma_{z} \otimes \mathbb{1} ) W_{\alpha} ( \sigma_{z} \otimes \mathbb{1} ) = - W_{\alpha},
\end{equation*}
the spectrum of $W_{\alpha}$ is symmetric and so $W_{\alpha}$ must have an eigenvalue of $+ \sqrt{ \lambda_{\gamma} }$. Choosing the corresponding eigenstate leads to a realisation satisfying $\beta_{\alpha} = \sqrt{ \lambda_{\gamma} }$. On the other hand, it is easy to check that for these observables
\begin{equation*}
T_{\alpha} = 2 \big[ ( \alpha^{2} - 1 ) \cos \gamma + 2 \alpha \sin \gamma \big] \cdot \mathbb{1}
\end{equation*}
and
\begin{equation*}
t_{\alpha} = \frac{\alpha^{2} - 1}{2} \cos \gamma + \alpha \sin \gamma - \frac{ \alpha^{2} - 1}{2}.
\end{equation*}
Is it easy to verify that for this setup $\beta_{\alpha}^{2} = 4 ( \alpha^{2} + t_{\alpha} )$. Taking $\gamma \in [0, \pi/2]$ covers the specified range of $t_{\alpha}$.
\end{proof}
\section{The Mermin-Ardehali-Belinskii-Klyshko family}
\label{app:mabk}
Recall that we denote the observables of the $k$-th party by $A_{0}^{k}$ and $A_{1}^{k}$. For $n \geq 2$ the Bell operator $W_{n}$ is defined recursively as
\begin{equation*}
W_{n} := \frac{1}{2} \big( W_{n - 1} + W_{n - 1}' \big) \otimes A_{0}^{n} + \frac{1}{2} \big( W_{n - 1} - W_{n - 1}' \big) \otimes A_{1}^{n},
\end{equation*}
where $W_{1} = A_{0}^{1}$ and the primed operators are obtained by a local exchange of observables at every party, i.e.~$A_{0}^{j} \leftrightarrow A_{1}^{j}$ for all $j$.
\begin{widetext}
\begin{prop}
Let $\cS_{n}^{o}$ ($\cS_{n}^{e}$) be the set of $n$-bit strings of odd (even) parity. For all $n \geq 1$ the following operator inequalities hold:
\begin{equation*}
\abs[\big]{ [ W_{n}, W_{n}' ] } \leq 2 \comsum{o}
\end{equation*}
and $W_{n}^{2} \leq R_{n}, W_{n}'^{2} \leq R_{n}$ for
\begin{equation*}
R_{n} := \comsum{e}.
\end{equation*}
\end{prop}
\begin{proof}
Proof by induction. It is clear that both inequalities hold for $n = 1$, so we only need to show that if the inequalities hold for $n$, then they also hold for $n + 1$.

Expanding the commutator $[ W_{n + 1}, W_{n + 1}' ]$ gives
\begin{equation*}
[ W_{n + 1}, W_{n + 1}' ] = \frac{1}{2} \big( W_{n}^{2} + W_{n}'^{2} \big) \otimes \big[ A_{0}^{n + 1}, A_{1}^{n + 1} \big] + \frac{1}{2} [ W_{n}, W_{n}' ] \otimes \big( (A_{0}^{n + 1})^{2} + (A_{1}^{n + 1})^{2} \big).
\end{equation*}
The matrix modulus of the commutator can be upperbounded by
\begin{align*}
\abs[\big]{[ W_{n + 1}, W_{n + 1}' ]} &\leq \frac{1}{2} \big( W_{n}^{2} + W_{n}'^{2} \big) \otimes \abs[\big]{\big[ A_{0}^{n + 1}, A_{1}^{n + 1} \big]} + \abs[\big]{[ W_{n}, W_{n}' ]} \otimes \mathbb{1}\\
&\leq R_{n} \otimes \abs[\big]{\big[ A_{0}^{n + 1}, A_{1}^{n + 1} \big]} + \abs[\big]{[ W_{n}, W_{n}' ]} \otimes \mathbb{1},
\end{align*}
which completes the proof of the first inequality.

Since $R_{n}$ is invariant under the exchange of observables, it suffices to prove that $W_{n + 1}^{2} \leq R_{n + 1}$. Expanding $W_{n + 1}^{2}$ gives
\begin{align*}
W_{n + 1}^{2} &= \frac{1}{4} \big( W_{n} + W_{n}' \big)^{2} \otimes \big( A_{0}^{n + 1} \big)^{2} + \frac{1}{4} \big( W_{n} - W_{n}' \big)^{2} \otimes \big( A_{1}^{n + 1} \big)^{2}\\
&+ \frac{1}{4} \big( W_{n}^{2} - W_{n}'^{2} \big) \otimes \{ A_{0}^{n + 1}, A_{1}^{n + 1} \} - \frac{1}{4} [ W_{n}, W_{n}'] \otimes [A_{0}^{n + 1}, A_{1}^{n + 1}].
\end{align*}
We upperbound the first two terms using $(A_{0}^{n + 1})^{2}, (A_{1}^{n + 1})^{2} \leq \mathbb{1}$ to obtain
\begin{equation*}
\frac{1}{4} \big( W_{n} + W_{n}' \big)^{2} \otimes \big( A_{0}^{n + 1} \big)^{2} + \frac{1}{4} \big( W_{n} - W_{n}' \big)^{2} \otimes \big( A_{1}^{n + 1} \big)^{2} \leq \frac{1}{2} \big( W_{n}^{2} + W_{n}'^{2} \big) \otimes \mathbb{1}.
\end{equation*}
Adding the third term and applying $W_{n}^{2}, W_{n}'^{2} \leq R_{n}$ gives
\begin{align*}
\frac{1}{2} \big( &W_{n}^{2} + W_{n}'^{2} \big) \otimes \mathbb{1} + \frac{1}{4} \big( W_{n}^{2} - W_{n}'^{2} \big) \otimes \{ A_{0}^{n + 1}, A_{1}^{n + 1} \}\\
&= \frac{1}{4} W_{n}^{2} \otimes \big( 2 \cdot \mathbb{1} + \{ A_{0}^{n + 1}, A_{1}^{n + 1} \} \big) + \frac{1}{4} W_{n}'^{2} \otimes \big( 2 \cdot \mathbb{1} - \{ A_{0}^{n + 1}, A_{1}^{n + 1} \} \big)\\
&\leq \frac{1}{4} R_{n} \otimes \big( 2 \cdot \mathbb{1} + \{ A_{0}^{n + 1}, A_{1}^{n + 1} \} \big) + \frac{1}{4} R_{n} \otimes \big( 2 \cdot \mathbb{1} - \{ A_{0}^{n + 1}, A_{1}^{n + 1} \} \big) = R_{n} \otimes \mathbb{1}.
\end{align*}
For the last term we use
\begin{equation*}
- \frac{1}{4} [ W_{n}, W_{n}'] \otimes [A_{0}^{n + 1}, A_{1}^{n + 1}] \leq \frac{1}{4} \abs{ [ W_{n}, W_{n}'] } \otimes \abs{ [A_{0}^{n + 1}, A_{1}^{n + 1}] }.
\end{equation*}
Combining all the terms we obtain
\begin{equation*}
W_{n + 1}^{2} \leq R_{n} \otimes \mathbb{1} + \frac{1}{4} \abs{ [ W_{n}, W_{n}'] } \otimes \abs{ [A_{0}^{n + 1}, A_{1}^{n + 1}] },
\end{equation*}
which completes the proof of the second inequality.
\end{proof}
\end{widetext}
\begin{prop}
\label{prop:mabk-tight-trade-off}
The trade-off between the Bell violation $\beta_{n}$ and the measure of incompatibility for the $k$-th party $t_{k}$
\tradeoffmabk{}
is tight for $t_{k} \in [0, 1]$.
\end{prop}
\begin{proof}
The proof relies on the fact that for projective observables $W_{n}^{2}$ admits a particularly simple form~[64]
\begin{equation*}
W_{n}^{2} = \sum_{x \in \cS_{n}^{e}} \bigotimes_{j = 1}^{n} \bigg( \frac{ i [ A_{0}^{j}, A_{1}^{j} ] }{2} \bigg)^{x_{j}}.
\end{equation*}
Consider the qubit observables
\begin{align*}
A_{0}^{j} &= \sigma_{x},\\
A_{1}^{j} &=
\begin{cases}
\cos \gamma \, \sigma_{x} + \sin \gamma \, \sigma_{y} \nbox{if} j = k,\\
\sigma_{y} \nbox{otherwise}
\end{cases}
\end{align*}
for $\gamma \in [0, \pi/2]$. It is easy to see that the largest eigenvalue of $W_{n}^{2}$ equals
\begin{equation*}
\lambda_{\gamma} = 2^{n - 2} ( 1 + \sin \gamma )
\end{equation*}
and corresponds to the subspace spanned by $\ket{0}^{\otimes n}$ and $\ket{1}^{\otimes n}$. Since the spectrum of $W_{n}$ is symmetric, it must have an eigenvalue of $+ \sqrt{ \lambda_{\gamma} }$. It is easy to check that $t_{k} = \sin \gamma$ and so choosing $\gamma \in [0, \pi/2]$ saturates the trade-off in the desired range.
\end{proof}
\section{Three anticommuting observables}
\label{app:three-observables}
\begin{prop}
Let $A_{0}, A_{1}, A_{2}$ be three binary observables and a let $\rho_{A}$ be a full-rank state such that
\begin{equation}
\label{eq:three-commutators}
\tr \big( \abs{ [ A_{j}, A_{k} ] } \rho_{A} \big) = 2
\end{equation}
for all $j \neq k$. Then, there exists a unitary $U_{A}$ such that
\begin{align*}
A_{0} &= U_{A} ( \sigma_{x} \otimes \mathbb{1} ) U_{A}^{\dagger},\\
A_{1} &= U_{A} ( \sigma_{y} \otimes \mathbb{1} ) U_{A}^{\dagger},\\
A_{2} &= U_{A} ( \sigma_{z} \otimes \Upsilon ) U_{A}^{\dagger},
\end{align*}
where $\Upsilon$ is a Hermitian projective observable, i.e.~$\Upsilon = \Upsilon^{\dagger}$ and $\Upsilon^{2} = \mathbb{1}$.
\end{prop}
\begin{proof}
Applying Proposition~\ref{prop:biased-chsh-optimal-observables} for $\alpha = 1$ to the pair $A_{0}, A_{1}$ yields a unitary $U_{A}$ such that
\begin{align*}
A_{0} = U_{A} ( \sigma_{x} \otimes \mathbb{1} ) U_{A}^{\dagger},\\
A_{1} = U_{A} ( \sigma_{y} \otimes \mathbb{1} ) U_{A}^{\dagger}.
\end{align*}
Hence, we only need to ensure the correct form of $A_{2}$. Since $\{ \mathbb{1}, \sigma_{x}, \sigma_{y}, \sigma_{z} \}$ constitute an operator basis for linear operators acting on $\amsbb{C}^{2}$, we can without loss of generality write
\begin{equation*}
U_{A}^{\dagger} A_{2} U_{A} = \mathbb{1} \otimes K_{0} + \sigma_{x} \otimes K_{x} + \sigma_{y} \otimes K_{y} + \sigma_{z} \otimes K_{z}
\end{equation*}
for some Hermitian operators $K_{0}, K_{x}, K_{y}, K_{z}$. We have shown before that if the commutator is maximal and the state is full-rank, then we must have $\abs{ [A_{j}, A_{k}] } = 2 \cdot \mathbb{1}$ and therefore $\{ A_{j}, A_{k} \} = 0$. On the other hand, computing the anticommutator directly gives
\begin{equation*}
U_{A}^{\dagger} \{ A_{0}, A_{2} \} U_{A} = \sigma_{x} \otimes K_{0} + \mathbb{1} \otimes K_{x},
\end{equation*}
which immediately implies that $K_{0} = K_{x} = 0$. An analogous argument applied to $A_{1}$ and $A_{2}$ gives $K_{y} = 0$. Therefore, we are left with
\begin{equation*}
U_{A}^{\dagger} A_{2} U_{A} = \sigma_{z} \otimes K_{z}.
\end{equation*}
Since the commutator being maximal implies that $A_{2}$ is projective, we have $K_{z}^{2} = \mathbb{1}$, which completes the proof.
\end{proof}
It is easy to see that this argument can be applied recursively to derive the exact form for an arbitrary number of observables. Let $A_{0}, A_{1}, \ldots, A_{n - 1}$ be $n$ observables such that for $j \neq k$
\begin{equation*}
\tr \big( \abs{ [ A_{j}, A_{k} ] } \rho_{A} \big) = 2,
\end{equation*}
where $\rho_{A}$ is full-rank. If $n$ is even, then there exists a unitary $U_{A}$ such that
\begin{equation*}
U_{A}^{\dagger} A_{j} U_{A} =
\begin{cases}
\sigma_{z}^{\otimes \lfloor j/2 \rfloor} \otimes \sigma_{x} \otimes \mathbb{1} \nbox[4]{for even} j,\\
\sigma_{z}^{\otimes \lfloor j/2 \rfloor} \otimes \sigma_{y} \otimes \mathbb{1} \nbox[4]{for odd} j,
\end{cases}
\end{equation*}
where the last identity operator acts on all the remaining registers. Clearly, we have recovered the well-known construction of $n$ anticommuting observables on $n/2$ qubits.

If $n$ is odd, the first $n - 1$ observables take the same form as before, while the last one is given by
\begin{equation*}
U_{A}^{\dagger} A_{n} U_{A} = \sigma_{z}^{\otimes (n - 1)/2} \otimes \Upsilon,
\end{equation*}
where $\Upsilon$ is an arbitrary projective observable. Hence, the ambiguity related to the free projective observable $\Upsilon$ appears only for odd values of $n$.
\bibliographystyle{alphaarxiv}
\bibliography{/home/jedrek/projekty/tex/library}
\end{document}